\documentclass[submission]{eptcs}

\usepackage{marvosym}
\providecommand{\event}{WPTE 2018}

\usepackage{makeidx}

\usepackage[utf8x]{inputenc}

\usepackage{amsfonts}
\usepackage{amssymb}
\usepackage{amsthm}
\usepackage{amsmath}
\usepackage{stmaryrd}
\usepackage[inline]{enumitem}

\usepackage{breakurl}             
\usepackage{underscore}           
\usepackage{mathpartir}
\usepackage{hyperref}

\theoremstyle{definition}
\newtheorem{example}{Example}[section]
\newtheorem{definition}{Definition}[section]
\newtheorem{theorem}{Theorem}[section]
\newtheorem{corollary}{Corollary}[section]

\newcommand{\Var}{\ensuremath{\mathit{Var}}}
\newcommand{\Cfg}{\ensuremath{\mathit{Cfg}}}
\newcommand{\Nat}{\ensuremath{\mathit{Nat}}}
\newcommand{\Term}{\ensuremath{\mathit{Term}}}
\newcommand{\SUM}{\ensuremath{\mathit{SUM}}}

\newcommand{\emptylist}{\ensuremath{\mathit{Nil}}}
\newcommand{\emptydict}{\ensuremath{\epsilon}}

\newcommand{\spaceVert}{\ensuremath{\; \vert \;}}
\newcommand{\bigTab}{\ensuremath{\;\;\;\;}}

\newcommand{\aptrans}{\ensuremath{\Rightarrow^{\forall}}}
\newcommand{\optrans}{\ensuremath{\Rightarrow^{\exists}}}
\newcommand{\TStrans}[1]{\ensuremath{\Rightarrow^{\mathcal{T}}_{#1}}}
\newcommand{\ETStrans}[1]{\ensuremath{\Rightarrow^{\mathcal{\ext(T)}}_{#1}}}

\newcommand{\cfg}[2]{\langle #1 \spaceVert #2 \rangle}

\newcommand{\caexpr}[1]{ \llfloor #1 \rrfloor }
\newcommand{\cbexpr}[1]{ \llceil #1 \rrceil }
\newcommand{\cstmt}[1]{ \llbracket #1 \rrbracket }
\newcommand{\idc}[1]{ [#1] }
\newcommand{\intc}[1]{ \overline{#1} }
\newcommand{\boolc}[1]{ \underline{#1} }
\newcommand{\lisymb}{ \rightsquigarrow }
\newcommand{\List}[1]{{#1}^*}
\newcommand{\ext}{\ensuremath{\theta}}
\newcommand{\truesymb}{\mathit{True}}
\newcommand{\falsesymb}{\mathit{False}}

\title{Reducing Total Correctness to Partial Correctness by a Transformation of the Language Semantics}
\author{Andrei-Sebastian~Buruiană
  \institute{Alexandru Ioan Cuza University \& Bitdefender}
  \email{sburuiana@bitdefender.com}
\and
Ștefan~Ciobâcă (\Letter)
\institute{Alexandru Ioan Cuza University}
\email{stefan.ciobaca@info.uaic.ro}
}
\def\titlerunning{Reducing Total Correctness to Partial Correctness}
\def\authorrunning{S. Buruiană \& Ș. Ciobâcă}

\begin{document}
\maketitle

\begin{abstract}

  We give a language-parametric solution to the problem of \emph{total
    correctness}, by automatically reducing it to the problem of
  \emph{partial correctness}, under the assumption that an expression
  whose value decreases with each program step in a well-founded order
  is provided. Our approach assumes that the programming language
  semantics is given as a rewrite theory. We implement a prototype on
  top of the RMT tool and we show that it works in practice on a
  number of examples.
  
\end{abstract}

\section{Introduction}
The line of work on reachability logic
(see~\cite{rosu-ellison-schulte-2010-amast,rosu-stefanescu-2012-oopsla,rosu-stefanescu-ciobaca-moore-2013-lics,stefanescu-ciobaca-mereuta-moore-serbanuta-rosu-2014-rta,stefanescu-park-yuwen-li-rosu-2016-oopsla})
proposes language-parametric verification tools for programs. We
continue this line of work by introducing a language-parametric total
correctness checker. Our checker works by reducing the problem of
total correctness to the problem of partial correctness by a
transformation of the semantics of the programming language.

A program is \emph{partially correct} if its output satisfies the
postcondition for all inputs on which it terminates. A program is
\emph{totally correct} if it terminates on all inputs and its output
satisfies the postcondition. Therefore, total correctness is usually
proven by splitting the problem into two parts: first establish
partial correctness by using various Hoare-like logics
(e.g.,~\cite{Cao2018,stefanescu-park-yuwen-li-rosu-2016-oopsla}), and
then establish termination using a specialized termination prover
(e.g.,~\cite{Giesl2017,ALARCON2007105}).

More rarely, logics that can directly prove total correctness
(e.g.,~\cite{key-ifm2017,da2016modular}) are used. However, recent
work in automated termination proving
(e.g.,~\cite{termination-cav2013,variance-popl2007,ramsey-tacas2013,t2-tacas2016,terminator-pldi2006})
shows that it is beneficial to use information obtained by proving
properties of a program (e.g., invariants) in the termination
argument. Most formal verification tools $V$ take a (possibly
annotated) program $P$ as input and return $V(P)$, which is \emph{yes}
if the verification is successful and \emph{no} if there is a
counterexample; additionally, because such problems are typically
undecidable, the verifier could return \emph{unknown} or it could loop
indefinitely. In this setting, if the programming language of $P$
changes (e.g., when a new language standard is published), the
verifier $V$ needs to be upgraded as well and also proved sound --
which may not be trivial. Another downside of this approach is that
the same verification techniques need to be implemented and proved
sound for all languages of interest.

In our line of work
(see~\cite{rosu-ellison-schulte-2010-amast,rosu-stefanescu-2012-oopsla,rosu-stefanescu-ciobaca-moore-2013-lics,stefanescu-ciobaca-mereuta-moore-serbanuta-rosu-2014-rta,stefanescu-park-yuwen-li-rosu-2016-oopsla}),
we propose to build language-parametric verifiers $V$: in this
parametric setting, $V$ takes as input both the (possibly annotated)
program $P$ and the operational semantics $S$ of the programming
language of $P$. Then $V(S, P)$ returns \emph{yes}, \emph{no} or
\emph{unknown} (or loops indefinitely), depending on the particular
property that it checks of the program $P$ in the operational
semantics $S$. The advantage of this approach is that the verifier is
proved sound once and can then be used for various programming
languages.

Reachability logic, which is a sound and relatively complete proof
system for partial correctness, was introduced
in~\cite{stefanescu-ciobaca-mereuta-moore-serbanuta-rosu-2014-rta}. For
a verifier $V$ that implements this logic, $V(S, P)$ checks whether
the (annotated) program $P$ is partially correct, when interpreted
using the operational semantics $S$. In the present article, we
propose to construct a language-parametric verifier $V_t(S, P)$ that
checks \emph{total correctness} of the program $P$ in the operational
semantics $S$. Our approach works by applying a transformation on $S$
and the program $P$. We develop and prove the soundness of a
transformation function $\theta$ such that $V_t(S, P) = V(\theta(S),
\theta(P)).$ This means that total correctness of the program $P$ in
the semantics $S$ is the same as partial correctness of the program
$\theta(P)$ in the semantics $\theta(S)$ and therefore the existing
partial correctness verifier can be used in conjunction with the
transformation $\theta$ to obtain a total correctness prover for any
language.

\begin{figure}[t]
\[\begin{array}{llr}
\mathit{Id} ::= & {\tt x} \mid {\tt y} \mid {\tt z} \mid \ldots & \textit{identifiers (program variables)} \\
\mathit{Int} ::= & 0, 1, -1, \ldots & \textit{integers} \\
\mathit{Bool} ::= & \truesymb \mid \falsesymb & \textit{booleans} \\
\mathit{AE} ::= & \mathit{Int} \mid \mathit{Id} \mid \mathit{AE} +
                  \mathit{AE} \mid \ldots & \textit{arithmetic expressions}\\
\mathit{BE} ::= & \mathit{Bool} \mid \mathit{AE} = \mathit{AE} \mid \mathit{AE} <
                  \mathit{AE} \mid \mbox{$\mathit{not}$ $\mathit{BE}$} \mid \ldots & \textit{boolean expressions}\\
\mathit{Stmt} ::= & \mathit{skip} & \textit{empty statement} \\
                  & \mid \mbox{$\mathit{Stmt}$; $\mathit{Stmt}$} &
                                                                   \textit{sequence
                                                                   of statements} \\
                  & \mid \mbox{$\mathit{Id}$ := $\mathit{AE}$} & \textit{assignment}\\
                  & \mid \mbox{while $\mathit{BE}$ do $\mathit{Stmt}$}
                         & \textit{while loop} \\
                  & \mid \mbox{if $\mathit{BE}$ then $\mathit{Stmt}$ else
                    $\mathit{Stmt}$} & \textit{conditional statement} \\
\end{array}\]
\caption{\label{fig:imp-abstract-syntax}The abstract syntax, in
  BNF-like notation, of the IMP language, which is used throughout the
  paper as a running example.}
\end{figure}

Our approach assumes that the operational semantics $S$ of the
language in question is given as a rewrite theory with rules of the
form
\begin{center}
 $l \Rightarrow r\mbox{ if } b$,
 \end{center} 
 where $l$ and $r$ are two terms representing program configurations
 and $b$ is a boolean constraint. For our running example, we use a
 simple imperative language that we call IMP (see, e.g.,
 \cite{winskel1993formal}), whose abstract syntax is presented in
 Figure~\ref{fig:imp-abstract-syntax}. IMP configurations are pairs
 $\cfg{c_1 \lisymb c_2 \lisymb \cdots \lisymb c_n \lisymb
   \emptylist}{\mathit{env}}$ where $c_1, c_2, \ldots, c_n$ is a list
 of expressions or statements that are to be evaluated/executed in
 order and $\mathit{env}$ is a map from program identifiers (program
 variables) to integers. The notation $\emptylist$ stands for the
 empty list. The semantics of IMP consists of rewrite rules like
\[\begin{array}{lclr}
\cfg{(v := i) \lisymb l }{\mathit{env}} & \Rightarrow &
\cfg{l}{\mathit{update}(v, i, \mathit{env})} & \mbox{ and }\\
\cfg{\mbox{(if $b$ then $s_1$ else $s_2$)} \lisymb
  l}{\mathit{env}} & \Rightarrow & \cfg{s_1 \lisymb
  l}{\mathit{env}}\mbox{ if $b = \truesymb$}, &
\end{array}\]
\noindent which define the meaning of all language operators. The two
rules above illustrate parts of the semantics of the assignment
statement and of the if-then-else statement, respectively. The full
details on the syntax and semantics of IMP are formally given in
Section~\ref{sec:prelim}. However, we note that it is possible to
faithfully model a variety of languages in this manner, as shown
in~\cite{DBLP:journals/iandc/SerbanutaRM09}. Given a language
semantics $S$ as a parameter, reachability logic (defined
in~\cite{stefanescu-ciobaca-mereuta-moore-serbanuta-rosu-2014-rta})
can prove sequents of the form
\[S \vdash l \land \phi_l \Rightarrow^{\forall} \exists \tilde{x}.(r \land \phi_r),\]
where $l$ and $r$ are configuration terms and $\phi_l, \phi_r$ are
constraints. The intuitive meaning of a sequent is that any instance
of the configuration $l$ satisfying constraint $\phi_l$ either
diverges (does not terminate) or it reaches (along any path, hence the
$\forall$) in a finite number of steps an instance of the
configuration $r$ satisfying constraint $\phi_r$ and agreeing with $l$
on all variables except $\tilde{x}$. The full syntax and semantics of
the sequents are presented formally in Section~\ref{sec:prelim}. Note
that such sequents subsume the notion of partial correctness. For
example, the partial correctness of the SUM program
\begin{verbatim}
s := 0
while not (m = 0) do s := s + m; m := m - 1
\end{verbatim}
\noindent is represented by the following partial correctness sequent
\[\begin{array}{l} S \vdash \cfg{\textit{SUM}}{\mathit{env}_1} \land
    \mathit{lookup}({\tt m},
    \mathit{env}_1) = z \land z \geq 0 \Rightarrow^{\forall} \\
    \qquad \exists \mathit{env}_2.(\cfg{\mathit{skip}}{\mathit{env}_2} \land \mathit{lookup}({\tt s},
    \mathit{env}_2) = z(z+1)/2),\end{array}\]
\noindent which is derivable using reachability logic. The sequent
states that if we run the $\textit{SUM}$ program in a configuration
where the environment $\mathit{env}_1$ maps the program identifier
${\tt m}$ to a positive integer $z$, then the program eventually
reaches a configuration where there is nothing left to execute (hence
the $\mathit{skip}$) and where the identifier ${\tt s}$ is mapped to
the sum of the first $z$ positive naturals. The sequent
\[\begin{array}{l} S \vdash \cfg{\textit{SUM}}{\mathit{env}_1} \land
\mathit{lookup}({\tt m}, \mathit{env}_1) = z \Rightarrow^{\forall} \\ \qquad
\exists \mathit{env}_2.(\cfg{\mathit{skip}}{\mathit{env}_2} \land \mathit{lookup}({\tt s},
\mathit{env}_2) = z(z+1)/2)\end{array}\]
\noindent is also derivable (note that the constraint $z \geq 0$ does
not appear anymore). The sequent is valid when interpreted in a
partial correctness sense, since the program loops forever when
$z < 0$. We propose a language transformation that builds an
artificial semantics $\ext(S)$ from the semantics $S$ by adding to the
configuration a parameter that decreases with each rewrite step. The
formal expression that is used for the parameter is a program variant
(i.e., an expression whose value decreases with each program
step). For example, the previously illustrated rewrite rules for the
assignment statement and respectively for the conditional statement
become:
\[\begin{array}{lclr}
    (\cfg{(v := i) \lisymb l }{\mathit{env}}, n) & \Rightarrow &
    (\cfg{l}{\mathit{update}(v, i, \mathit{env})}, n - 1) & \mbox{ and }\\
    (\cfg{\mbox{(if $b$ then $s_1$ else $s_2$)} \lisymb
    l}{\mathit{env}}, n) & \Rightarrow & (\cfg{s_1 \lisymb
    l}{\mathit{env}}, n - 1) \mbox{ if $b = \truesymb$}. &
  \end{array}\]
In the new semantics, $\ext(S)$, all programs terminate, since the
variant is in a well-founded order and therefore it cannot decrease
indefinitely. Therefore, in order to prove total correctness of a
program $P$ in $S$, it is sufficient to prove partial correctness of
$(P, B)$ in $\ext(S)$, where $B$ is a sufficiently large bound. For
our running example, we can establish that
\[\begin{array}{l}
\ext(S) \vdash (\cfg{\textit{SUM}}{\mathit{env}_1}, 200|z| + 200)
\land \mathit{lookup}({\tt m}, \mathit{env}_1) = z \land z \geq 0
\Rightarrow^{\forall} \\ \qquad \exists g,\mathit{env}_2.((\cfg{\mathit{skip}}{\mathit{env}_2}, g)
\land \mathit{lookup}({\tt s}, \mathit{env}_2) = z(z+1)/2),\end{array}\]
which implies by our soundness theorem that \textit{SUM} is totally
correct, under the precondition that the program variable ${\tt m}$
starts with a nonnegative value. We have chosen the upper bound
$200|n| + 200$, since it is sufficiently large to allow for the
program to finish. The variable $g$ captures the number of execution
steps remaining from the initial $200|n| + 200$ steps. The sequent
above can be proven automatically (by relying on an invariant-like
annotation for the while loop) in our implementation. However, by our
soundness theorem, there is no bound $B$ such that
\[\begin{array}{l}
\ext(S) \vdash (\cfg{\textit{SUM}}{\mathit{env}_1}, B)
\land \mathit{lookup}({\tt m}, \mathit{env}_1) = z \Rightarrow^{\forall}
\\ \qquad \exists g,\mathit{env}_2.((\cfg{\mathit{skip}}{\mathit{env}_2}, g) \land
\mathit{lookup}({\tt s}, \mathit{env}_2) = z(z+1)/2),\end{array}\]
\noindent meaning that it is impossible to prove the total correctness
of the program $\textit{SUM}$ if there is no precondition for the
initial value of the program variable ${\tt m}$.

In contrast with some other automated termination provers (discussed
in Section~\ref{sec:related}), our method requires to provide the
upper bound on the number of steps manually. In the example above, we
picked $200|n| + 200$ because, intuitively, the program has a
linear-time complexity. The constant $200$ should be large enough to
allow the program to terminate.

The advantage and novelty of our method is that it is
language-parametric (the semantics of the language is given as an
input to our reduction). The main technical difficulties are to find a
sound but general enough transformation $\theta$ (given in
Definition~\ref{def:transformation}) and the right statement of the
soundness theorem (Theorem~\ref{th:main}).

\paragraph{Contributions.} \begin{enumerate}

\item We propose a language-parametric method of proving total
  correctness;

\item Our approach works by reducing total correctness to partial
  correctness using a language transformation and therefore it can
  also be seen as an argument for semantics-parametric program
  verifiers;

\item We implement the reduction in the RMT~\cite{rmt} tool and we use
  it to prove several interesting examples.

\end{enumerate}

\paragraph{Organization.}

In Section~\ref{sec:prelim}, we briefly introduce our notations for
many-sorted algebras and we recall matching logic and reachability
logic, which are the formalisms that we use to define and reason about
the operational semantics of languages. In
Section~\ref{sec:reduction}, we present our transformation, which
reduces total correctness to partial correctness, we prove its
soundness and we present the main
difficulties. Section~\ref{sec:related} discusses related work and
Section~\ref{sec:conclusion} concludes the paper, including possible
directions for future work.

\section{Preliminaries: Proving Partial Correctness using Reachability
  Logic}
\label{sec:prelim}
This section fixes notations for many-sorted algebra and recalls
matching logic and reachability logic.
\begin{figure}[t]

  \begin{center}
  \begin{tabular}{cc}
    \begin{tabular}{lll}
      $\intc{\;\cdot\;}$ & : & $\mathit{Int} \rightarrow \mathit{AE}$ \\
      $\idc{\cdot}$ & : & $\mathit{Id} \rightarrow \mathit{AE}$ \\
      $\mathit{plus}$ & : & $\mathit{AE} \times \mathit{AE} \rightarrow \mathit{AE}$ \\

      $\boolc{\; \cdot \;}$ & : & $\mathit{Bool} \rightarrow \mathit{BE}$\\
      $\mathit{eq}$ & : & $\mathit{AE} \times \mathit{AE} \rightarrow \mathit{BE}$ \\
      $\mathit{not}$ & : & $\mathit{BE} \rightarrow \mathit{BE}$ \\

      $\mathit{assign}$ & : & $\mathit{Id} \times \mathit{AE} \rightarrow \mathit{Stmt}$ \\
      $\mathit{seq}$ & : & $\mathit{Stmt} \times \mathit{Stmt} \rightarrow \mathit{Stmt}$ \\
      $\mathit{ite}$ & : & $\mathit{BE} \times \mathit{Stmt} \times \mathit{Stmt} \rightarrow \mathit{Stmt}$ \\
      $\mathit{while}$ & : & $\mathit{BE} \times \mathit{Stmt} \rightarrow \mathit{Stmt}$ \\
      $\mathit{skip}$ & : & $\rightarrow \mathit{Stmt}$ \\
      
      \hline
      
      $\caexpr{\cdot}$ & : & $\mathit{AE} \rightarrow \mathit{Code}$ \\
      $\cstmt{\cdot}$ & : & $\mathit{Stmt} \rightarrow \mathit{Code}$ \\
      $\cbexpr{\cdot}$ & : & $\mathit{BE} \rightarrow \mathit{Code}$ \\
      $\emptylist$ & : & $\rightarrow \mathit{Stack}$ \\
      $\cdot \lisymb \cdot$ & : & $\mathit{Code} \times \mathit{Stack} \rightarrow \mathit{Stack}$ \\
    \end{tabular}&
    \begin{tabular}{lll}
      $\emptydict$ & : & $\rightarrow \mathit{Env}$ \\
      $\mathit{\cfg{\cdot}{\cdot}}$ & : & $\mathit{Stack} \times \mathit{Env} \rightarrow \mathit{Cfg}$ \\
      \hline
      $\mathit{isInt}$ & : & $\mathit{AE} \!\rightarrow\! \mathit{Bool}$ \\
      $\mathit{isBool}$ & : & $\mathit{BE} \!\rightarrow\! \mathit{Bool}$\\
      $\cdot + \cdot$ & : & $\mathit{Int} \!\times\! \mathit{Int} \!\rightarrow\! \mathit{Int} $\\
      $\cdot = \cdot$ & : & $\mathit{Int} \!\times\! \mathit{Int} \!\rightarrow\! \mathit{Bool}$\\
      $! \cdot$ & : & $\mathit{Bool} \!\rightarrow\! \mathit{Bool}$\\
      $\mathit{lookup}$ & : & $\mathit{Id} \!\times\! \mathit{Env} \!\rightarrow\! \mathit{Int}$ \\
      $\mathit{update}$ & : & $\mathit{Id} \!\times\! \mathit{Int} \!\times\! \mathit{Env} \!\rightarrow\! \mathit{Env}$ \\
      \hline
      $\mathit{plushl}$ & : & $\mathit{AE} \!\rightarrow\! \mathit{AE}$ \\
      $\mathit{plushr}$ & : & $\mathit{AE} \!\rightarrow\! \mathit{AE}$ \\

      $\mathit{eqhl}$ & : & $\mathit{BE} \!\rightarrow\! \mathit{BE}$\\
      $\mathit{eqhr}$ & : & $\mathit{BE} \!\rightarrow\! \mathit{BE}$\\
      $\mathit{noth}$ & : & $\rightarrow\! \mathit{BE}$\\
      $\mathit{assignh}$ & : & $\mathit{Id} \!\rightarrow\! \mathit{Stmt}$ \\
      $\mathit{iteh}$ & : & $\mathit{Stmt} \!\times\! \mathit{Stmt} \!\rightarrow\! \mathit{Stmt}$ \\
    \end{tabular}
  \end{tabular}
  \end{center}
  \caption{\label{fig:function-symbols} The symbols in the signature
    $\Sigma$ used in our running example. For the infixed symbols, a
    centered dot represents an argument.}
\end{figure}
We denote by $\List{S}$ the set of ordered tuples, possibly empty,
with elements in $S$; we say that a set $T$ is \textit{$S$-indexed} if
$T = \left\{T_s \spaceVert s \in S \right\}$ is a collection of sets,
each one corresponding to a different item in $S$. For ease of
notation, we sometimes write $x \in T$ instead of $x \in T_s$ when $s$
is clear from context. A \emph{many-sorted signature} $\Sigma$ is an
ordered pair $\Sigma = (S, F)$, where $S$ is the set of sorts, and $F
= \left\{ F_{w,s} \spaceVert w \in \List{S}, s \in S \right\}$ is the
$(\List{S} \times S)$-indexed set of function symbols. If $f\in
F_{w,s}$, we say that $f$ is a \emph{function symbol} of arity
$(w,s)$. If $w = (s_1,\ldots,s_n)$, we sometimes write $f :
s_1,\ldots,s_n \rightarrow s$ instead of $f\in F_{w,s}$, which
indicates that the function symbol $f$ has arguments of sorts
$s_1,\ldots,s_n$ and a result of sort $s$ ($f \in F_{w, s}$).

A $\Sigma$-algebra is a pair $\mathcal{A} = (A, I_A)$, where $A =
\left\{ A_s \spaceVert s \in S \right\}$ is an $S$-indexed set called
the carrier set of $\mathcal{A}$ and $I_A(f)$ is a function, $I_A(f) :
A_{s_1} \times \ldots \times A_{s_n} \rightarrow A_{s}$, for all $f
\in F_{(s_1,\ldots,s_n),s}$. That is, the interpretation map $I_A$
assigns to each function symbol in $F$ a function of the appropriate
arity. For convenience, we sometimes refer to the algebra
$\mathcal{A}$ as a set, in which case we mean its carrier set $A$. We
assume as usual that $A_s \not= \emptyset$ for any $s \in S$. Given an
$S$-indexed set of symbols $\Var$, we denote by $\Term_{\Sigma,
  s}(\Var)$ the set of terms of sort $s$ built with function symbols
in $\Sigma$ and variables in $\Var$ and by $\Term_{\Sigma}(\Var)$ the
$S$-indexed set of all terms with variables in $\Var$. Given a
$\Sigma$-algebra $\mathcal{A}$ with carrier set $A = \left\{A_s
\spaceVert s \in S\right\}$, \emph{a valuation} $\rho : \bigcup_{s \in
  S}\Var_s \rightarrow \bigcup_{s \in S}A_s$ is a function that
assigns to each variable an element in $A$ of the appropriate
sort. Valuations extend homomorphically to terms as usual.
We now recall \textit{matching logic}, as introduced
in~\cite{stefanescu-ciobaca-mereuta-moore-serbanuta-rosu-2014-rta}. Fix
an algebraic signature $\Sigma = (S, F)$ with a distinguished sort
$\Cfg \in S$ called the sort of configurations, an $S$-indexed set of
variables $\Var$ and a $\Sigma$-algebra $\mathcal{T}$ with carrier set
$T$. $T$ is called the \textit{configuration model}. The elements of
the algebra $\mathcal{T}$ of sort $\Cfg$, denoted by
$\mathcal{T}_{\Cfg}$, are called \textit{configurations}. Matching
logic is a logic of program configurations.

\begin{example}
We consider a running example where the elements of $\mathcal{T}$ of
sort $\Cfg$ are programs, running in an environment, written in a
simple imperative language that we call IMP. We work in the signature
$(S, \Sigma)$, where $S = \{\mathit{Int}, \mathit{Bool}, \mathit{AE},
\mathit{BE}, \mathit{Id}, \mathit{Stmt}, \mathit{Stack}, \mathit{Env},
\mathit{Cfg}, \mathit{Code} \}$ and where the function symbols in
$\Sigma$ are presented in Figure~\ref{fig:function-symbols}. The first set of symbols is used to represent the syntax of IMP programs. The second set of symbols is required to represent configurations, which consist of a stack of code to be executed/evaluated, and an environment
mapping identifiers to integers. The third set of symbols
represents mathematical operations. The last set consists of
several auxiliary symbols, which are necessary to specify the
rules of the operational semantics. For brevity, not all operators
are presented; there are additional operations for less-than,
boolean connectives, etc.
The sorts $\mathit{Int}$ and $\mathit{Bool}$ are interpreted by
mathematical integers and booleans, respectively. The sorts
$\mathit{AE}$, $\mathit{BE}$ and $\mathit{Stmt}$ are the sorts for
arithmetic expressions, boolean expressions and statements,
respectively. The sort $\mathit{Id}$ is for program identifiers
(program variables). There are injections $\caexpr{\cdot}$,
$\cbexpr{\cdot}$ and $\cstmt{\cdot}$ from $\mathit{AE}$, $\mathit{BE}$
and $\mathit{Stmt}$, respectively, into the sort
$\mathit{Code}$. Therefore $\mathit{Code}$ refers to either arithmetic
or boolean expressions, or statements.  $\mathit{Env}$ is the sort of
maps from $\mathit{Id}$s to $\mathit{Int}$egers. The sort
$\mathit{Stack}$ refers to a stack of $\mathit{Code}$s that should be
evaluated/executed in order, starting with the top of the
stack. Configurations (of sort $\mathit{Cfg}$) consist of a
$\mathit{Stack}$ and of an environment of sort $\mathit{Env}$. The
symbols in the signature $\Sigma$ are presented in
Figure~\ref{fig:function-symbols}. It includes all function symbols
needed to represent the initial configuration, but also helper symbols
that occur during program execution.
\end{example}

\begin{example}
The SUM program introduced earlier, placed in an initial configuration
with the empty environment, $\emptydict$, is represented by the
following term of sort $\mathit{Cfg}$:
  \[\begin{array}{l}\cfg{\cstmt{\mathit{seq}(\mathit{assign(s, \intc{\mbox{$0$}})},\\
      \phantom{\langle \llbracket \mathit{seq}(}\mathit{while}(not(eq(\intc{0}, \idc{m})),
      \mathit{seq}(\mathit{assign}(s, \mathit{plus}(\idc{s},
      \idc{m})), \\
      \phantom{\langle \llbracket \mathit{seq}(\mathit{while}(not(eq(\intc{0}, \idc{m})),
      \mathit{seq}(}\mathit{assign}(m, \mathit{plus}(\idc{m},
      \intc{-1})) )))} \lisymb \emptylist}{\emptydict}.\end{array}\]
	
\end{example}

The rest of this section recalls definitions
from~\cite{stefanescu-ciobaca-mereuta-moore-serbanuta-rosu-2014-rta}.

\begin{definition}
	A \emph{matching logic formula} (or \emph{pattern}), is a
        first-order logic (FOL) formula that additionally allows terms
        in $\Term_{\Sigma, \Cfg}(\Var)$, called \emph{basic patterns},
        as atomic formulae. We recall that by $\Term_{\Sigma,
          \Cfg}(\Var)$ we denote the terms of sort $\Cfg$ in the
        $\Sigma$-algebra of terms. We say that a pattern is
        \emph{structureless} if it contains no basic patterns. More
        formally, a matching logic formula is defined as follows:
	\begin{enumerate}
	\item if $\pi \in \Term_{\Sigma, \Cfg}(\Var)$, then $\pi$ is a formula;
	\item if $w = (s_1, \ldots, s_n)$, $t_i \in \Term_{\Sigma, s_i}(\Var)$ for all $i \in \left\{1, \ldots, n\right\}$ and $P \in F_{w, \mathit{Bool}}$, then $P(t_1, \ldots, t_n)$ is a formula;
	\item if $\varphi_1$ and $\varphi_2$ are formulae, then $\varphi_1 \wedge \varphi_2$ and $\varphi_1 \vee \varphi_2$ are formulae;	
	\item if $\varphi$ is a formula, then $\neg \varphi$ is a formula;
	\item if $\varphi$ is a formula and $x \in \Var$, then $\exists x \varphi$ and $\forall x \varphi$ are formulae.
	\end{enumerate}
\end{definition}
By $\mathcal{P}_{\mathcal{T}}$ we denote the set of all patterns over an algebra $\mathcal{T}$.

\begin{definition}
  For a fixed algebra $\mathcal{T} = (A, I)$, we define \emph{satisfaction} $(\gamma,\rho) \models \varphi$ over configurations $\gamma \in \mathcal{T}_{\Cfg}$, valuations $\rho:\Var \rightarrow \mathcal{T}$ and patterns $\varphi$ as follows:
  \begin{enumerate}
	
	\item $(\gamma,\rho) \models$ $P(t_1, t_2, \ldots, t_n)$ if and only if $(I(P))(\rho(t_1), \rho(t_2), \ldots, \rho(t_n)) = \top$;
	
	\item $(\gamma,\rho) \models \pi$ iff $\gamma = \rho(\pi)$ where $\pi \in \Term_{\Sigma, \Cfg}(\Var)$;
	
	\item $(\gamma,\rho) \models (\varphi_1 \wedge \varphi_2)$ iff $(\gamma,\rho) \models \varphi_1$ and $(\gamma,\rho) \models \varphi_2$;
	
	\item $(\gamma,\rho) \models (\varphi_1 \vee \varphi_2)$ iff $(\gamma,\rho) \models \varphi_1$ or $(\gamma,\rho) \models \varphi_2$;
	
	\item $(\gamma,\rho) \models \neg \varphi$ iff $(\gamma,\rho) \; \not\models \; \varphi$;
	
	\item $(\gamma,\rho) \models \exists X \varphi$ iff $(\gamma,
          \rho') \models \varphi$ for some $\rho' : \Var \rightarrow \mathcal{T}$ with $\rho'(y) = \rho(y)$ for all $y \in Var \backslash \{X\}$;

	\item $(\gamma,\rho) \models \forall X \varphi$ iff $(\gamma,\rho) \; \not\models \; \exists X (\neg \varphi)$.

  \end{enumerate}
	
	We write $\models \varphi$ when $(\gamma, \rho) \models \varphi$ for all $\gamma \in \mathcal{T}_{\Cfg}$ and all $\rho:\Var \rightarrow \mathcal{T}$.
\end{definition}

We now recall all-path reachability logic (as presented
in~\cite{stefanescu-ciobaca-mereuta-moore-serbanuta-rosu-2014-rta}).

\begin{definition}
	A (one-path) \emph{reachability rule} is an ordered pair of patterns $(\varphi, \varphi ')$ (which can have free variables). We write this pair as $\varphi \Rightarrow ^ {\exists} \varphi '$. We say that rule $\varphi \Rightarrow ^ {\exists} \varphi '$ is \emph{weakly well-defined} iff for any $\gamma \in \mathcal{T}_{\Cfg}$ and $\rho : \Var \rightarrow \mathcal{T}$ with $(\gamma, \rho) \models \varphi$, there exists $\gamma ' \in \mathcal{T}_{\Cfg}$ such that $(\gamma ', \rho) \models \varphi '$.
\end{definition}

\begin{definition}
	A \emph{reachability system} is a set of reachability rules. A reachability system $S$ is \emph{weakly well-defined} iff each rule is weakly well-defined. $S$ induces a \emph{transition system} $(\mathcal{T}, \Rightarrow^{\mathcal{T}}_S)$ on the configuration model: $\gamma \Rightarrow^{\mathcal{T}}_S \gamma '$ for $\gamma, \gamma ' \in \mathcal{T}_{\Cfg}$ iff there is some rule $\varphi \Rightarrow^{\exists} \varphi ' \in S$ and some valuation $\rho : \Var \rightarrow \mathcal{T}$ such that $(\gamma, \rho) \models \varphi$ and $(\gamma ', \rho) \models \varphi '$. We write $\Rightarrow$ instead of $\Rightarrow^{\mathcal{T}}_S$ when it is clear from context that we are referring to a particular transition system.
\end{definition}

\begin{example}
We consider a fixed $\Sigma$-algebra $\mathcal{T}$ having the
following properties: $\mathcal{T}_{\mathit{Int}} = \mathbb{Z}$,
$\mathcal{T}_{\mathit{Bool}} = \{ \truesymb, \falsesymb \}$,
$\mathcal{T}_{\mathit{Id}} = \{x, y, z\ldots\}$, $\mathcal{T}_{a + b}
= a + b$ \textit{for all} $a, b \in \mathbb{Z}$, $\mathcal{T}_{lookup(X, update(X, I, \mathit{env}))} =
I$ \textit{for all} $X \in \mathcal{T}_{\mathit{Id}}, I \in
\mathbb{Z}$, $\mathit{env} \in \mathcal{T}_{\mathit{Env}}$,
$\mathcal{T}_{\mathit{lookup}(Y, update(X, I, \mathit{env}))} =
\mathcal{T}_{\mathit{lookup}(Y, \mathit{env})}$ \textit{for all} $Y \in
\mathcal{T}_{\mathit{Id}} \setminus \{ X \}, I \in \mathbb{Z}$,
$\mathit{env} \in \mathcal{T}_{\mathit{Env}}$, $\mathcal{T}_{\mathit{lookup}(X,
  \emptydict)} = 0$ \textit{for all} $X \in \mathcal{T}_{\mathit{Id}}$,
$\mathcal{T}_{\mathit{isInt}(x)} = \truesymb$ iff $x = \intc{y}$, for
some $y \in \mathbb{Z}$ and $\mathcal{T}_{\mathit{isBool}(x)} =
\truesymb$ iff $x = \boolc{y}$, for some $y \in \mathcal{T}_{\mathit{Bool}}$. The
weakly well-defined system $S$ defining the operational semantics of
IMP is presented in Figure~\ref{fig:imp-semantics}. For brevity,
some rules that are similar to existing rules are missing (e.g.,
the rules for $\mathit{eq}$ are similar to those for
$\mathit{plus}$). We discuss the first four rules, which define
the assignment operator and the lookup. The first rule schedules
the expression on the rhs of an assignment to be evaluated, if it
is not already an integer. Once the expression is evaluated to an
integer (using the other rules), the second rule places the result
back into the assignment operator. Once the rhs is an integer, the
third rule updates the environment appropriately. The fourth rule
evaluates a variable by looking it up in the environment. The reachability
system $S$ generates the transition relation
$\Rightarrow^{\mathcal{T}}_S$ on the model $\mathcal{T}$. Note that
reachability rules of the form $l \land \phi \Rightarrow^\exists r$
(with $l \land \phi$ and $r$ being matching logic
formulae) subsume the rewrite rules of the form $l \Rightarrow r\mbox{
  if }\varphi$ used in the introduction.

\begin{figure}[t]
  \centering
  \begin{tabular}{l}
    $\cfg{ \cstmt{\mathit{assign}(X, A)} \lisymb T }{ \mathit{env}}
    \wedge \neg\mathit{isInt(A)} \Rightarrow^{\exists} \cfg{
      \caexpr{A} \lisymb \cstmt{\mathit{assignh}(X)} \lisymb T }{
      \mathit{env}}$ \\
    
    $ \cfg{ \caexpr{\intc{a}} \lisymb \cstmt{\mathit{assignh}(X)}
      \lisymb T }{ \mathit{env}} \Rightarrow^{\exists} \cfg{
      \cstmt{\mathit{assign}(X, \intc{a})} \lisymb T }{
      \mathit{env}}$\\
    
    $\cfg{ \cstmt{\mathit{assign}(X, \intc{a})} \lisymb T }{
      \mathit{env}} \Rightarrow^{\exists} \cfg{ T }{
      \mathit{update}(X, a, \mathit{env})}$ \\
    
    $\cfg{ \caexpr{\mathit{\idc{X}}} \lisymb T }{ \mathit{env}} \Rightarrow^{\exists} \cfg{
      \caexpr{\intc{\mathit{lookup}(\mathit{X}, \mathit{env})}}
      \lisymb T }{ \mathit{env}}$ \\
    
    $\cfg{ \cstmt{skip} \lisymb T }{ \mathit{env}}
    \Rightarrow^{\exists} \cfg{ T }{ \mathit{env}}$ \\
    
    $\cfg{ \cstmt{\mathit{seq}(S_1, S_2)} \lisymb T }{ \mathit{env}}
    \Rightarrow^{\exists} \cfg{ \cstmt{S_1} \lisymb \cstmt{S_2} \lisymb T }{ \mathit{env}}$ \\
    
    $\cfg{ \cstmt{\mathit{ite}(\cbexpr{\boolc{\falsesymb}}, S_1, S_2)} \lisymb T }{ \mathit{env}}
    \Rightarrow^{\exists} \cfg{ \cstmt{S_2} \lisymb T }{ \mathit{env}}$ \\
    
    $\cfg{ \cstmt{\mathit{ite}(\cbexpr{\boolc{\truesymb}}, S_1, S_2)} \lisymb T }{ \mathit{env}}
    \Rightarrow^{\exists} \cfg{ \cstmt{S_1} \lisymb T }{ \mathit{env}}$ \\
    
    $\cfg{ \cstmt{\mathit{ite}(C, S_1, S_2)} \lisymb T }{ \mathit{env}} \wedge \neg\mathit{isBool}(C)
    \Rightarrow^{\exists} $ \\
    
    $\qquad \cfg{ \cbexpr{C} \lisymb \cstmt{\mathit{iteh}(S_1, S_2)} \lisymb T }{ \mathit{env}}$ \\
    
    $\cfg{ \cbexpr{C} \lisymb \cstmt{\mathit{iteh}(S_1, S_2)} \lisymb T }{ \mathit{env}} \wedge \mathit{isBool}(C)
    \Rightarrow^{\exists}$ \\
    
    $\qquad \cfg{ \cstmt{\mathit{ite}(C, S_1, S_2)} \lisymb T }{ \mathit{env}}$ \\
    
    $\cfg{ \cstmt{\mathit{while}(C, S)} \lisymb T }{ \mathit{env}}
    \Rightarrow^{\exists} \cfg{ \cstmt{\mathit{ite}(C, \mathit{seq}(S, \mathit{while}(C, S)), \mathit{ skip})} \lisymb T }{ \mathit{env}}$ \\
    
    $\cfg{ \caexpr{\mathit{plus}(\intc{a}, \intc{b})} \lisymb T }{ \mathit{env}}
    \Rightarrow^{\exists} \cfg{ \caexpr{\intc{a + b}} \lisymb T }{ \mathit{env}}$ \\
    
    $\cfg{ \caexpr{\mathit{plus}(A, B)} \lisymb T }{ \mathit{env}} \wedge \neg\mathit{isInt}(A)
    \Rightarrow^{\exists} \cfg{ \caexpr{A} \lisymb \caexpr{\mathit{plushl}(B)} \lisymb T }{ \mathit{env}}$ \\
    
    $\cfg{ \caexpr{\mathit{plus}(A, B)} \lisymb T }{ \mathit{env}} \wedge \mathit{isInt}(A) \wedge \neg\mathit{isInt}(B)
    \Rightarrow^{\exists}$ \\
    
    $\qquad \cfg{ \caexpr{B} \lisymb \caexpr{\mathit{plushr}(A)} \lisymb T }{ \mathit{env}}$ \\
    
    $\cfg{ \caexpr{A} \lisymb \caexpr{\mathit{plushl}(B)} \lisymb T }{ \mathit{env}} \wedge \mathit{isInt}(A)
    \Rightarrow^{\exists} \cfg{ \caexpr{\mathit{plus}(A, B)} \lisymb T }{ \mathit{env}}$ \\
    
    $\cfg{ \caexpr{B} \lisymb \caexpr{\mathit{plushr}(A)} \lisymb T }{ \mathit{env}} \wedge \mathit{isInt}(B)
    \Rightarrow^{\exists} \cfg{ \caexpr{\mathit{plus}(A, B)} \lisymb T }{ \mathit{env}}$ \\
    
  \end{tabular}
  \caption{\label{fig:imp-semantics}The reachability system $S$
    defining the semantics of IMP. Capital letters represent variables
    of the appropriate sorts. The variables $a, b$ stand for integers
    and the variable $\mathit{env}$ for an environment.}
\end{figure}

\end{example}

\begin{definition}
  A $\Rightarrow^{\mathcal{T}}_S$-\emph{execution} is a sequence
  $\gamma_0 \Rightarrow^{\mathcal{T}}_S \gamma_1
  \Rightarrow^{\mathcal{T}}_S \cdots$, potentially infinite, where
  $\gamma_0, \gamma_1, \ldots \in \mathcal{T}_{\Cfg}$. If a
  $\Rightarrow^{\mathcal{T}}_S$-execution is finite, we call it a
  $\Rightarrow^{\mathcal{T}}_S$-\emph{path}. We say that such a path
  is \emph{complete} iff it is not a strict prefix of any other
  $\Rightarrow^{\mathcal{T}}_S$-path (i.e., the last element is
  irreducible).
\end{definition}

The following is an example of a complete $\TStrans{S}$-path:
\begin{center}
	$\cfg{\cstmt{\mathit{seq}(\mathit{skip}, \mathit{skip})} \lisymb \emptylist}{\emptydict} \Rightarrow \cfg{\cstmt{\mathit{skip}} \lisymb \cstmt{\mathit{skip}} \lisymb \emptylist}{\emptydict} \Rightarrow
	\cfg{\cstmt{\mathit{skip}} \lisymb \emptylist}{\emptydict} \Rightarrow
	\cfg{\emptylist}{\emptydict}$.
\end{center}

\begin{definition}[Partial Correctness]
  \label{def:partial-correctness}
  An \emph{all-path reachability rule} is a pair $\varphi \aptrans
  \varphi'$. We say that $\varphi \aptrans \varphi'$ is
  \emph{satisfied} by $S$, denoted by $S \models \varphi \aptrans
  \varphi'$, iff for all complete $\TStrans{S}$-paths $\tau$ starting
  with $\gamma \in \mathcal{T}_{\Cfg}$ and for all $\rho : \Var
  \rightarrow \mathcal{T}$ such that $(\gamma, \rho) \models \varphi$,
  there exists some $\gamma' \in \tau$ such that $(\gamma', \rho)
  \models \varphi'$.
\end{definition}

The definition above generalizes typical partial correctness of Hoare
tuples of the form $\{ \varphi \} \mathit{P} \{ \varphi' \}$, as the
reachability formula $\mathit{P} \land \varphi \Rightarrow^\forall
\cfg{\mathit{skip}}{\mathit{env}} \land \varphi'$ can be used
instead. See~\cite{stefanescu-ciobaca-mereuta-moore-serbanuta-rosu-2014-rta}
for a more detailed discussion. Reachability logic has a sound and
relatively complete proof system, which derives sequents of the form
$S \vdash \varphi \Rightarrow^\forall \varphi'$ if and only if $S
\models \varphi \Rightarrow^\forall \varphi'$ holds. The results in
the present paper do not depend on the proof system, and therefore the
proof system is presented in Appendix~\ref{app:proofsystem}.

\section{The Reduction of Total Correctness to Partial Correctness}

\label{sec:reduction}

We now present a transformation that reduces total correctness to the
problem of partial correctness. We first define what it means for a
pattern to terminate.

\begin{definition}[Termination of a Pattern]
  We say that a pattern $\varphi$ \emph{terminates} in $S$ if for all
  $\gamma \in \mathcal{T}_{\Cfg}$ and all $\rho : \Var \rightarrow
  \mathcal{T}$ such that $(\gamma, \rho) \models \varphi$, all
  executions $\gamma \Rightarrow \gamma_1 \Rightarrow \gamma_2
  \Rightarrow \cdots$ from $\gamma$ in $(\mathcal{T},
  \Rightarrow^{\mathcal{T}}_{S})$ are finite.
\end{definition}

\begin{example}

  The following pattern does \emph{not} terminate in
  $S$: \[\cfg{\cstmt{\mathit{while}(C, \mathit{skip})} \lisymb
    \emptylist}{\emptydict}\mbox{, where }C \in \Var_{\mathit{BE}}.\]
  Its nontermination is witnessed by the following execution:
  \[\begin{array}{c}
    \cfg{\cstmt{\mathit{while}(\cbexpr{\boolc{\truesymb}}, \mathit{skip})} \lisymb \emptylist}{\emptydict} \TStrans{S} \\
    \cfg{ \cstmt{\mathit{ite}(\cbexpr{\boolc{\truesymb}}, \mathit{seq}(\mathit{skip}, \mathit{while}(\cbexpr{\boolc{\truesymb}}, \mathit{skip})), \mathit{skip})} \lisymb T }{ \emptydict}  \TStrans{S} \cdots \\
    \cfg{\cstmt{\mathit{while}(\cbexpr{\boolc{\truesymb}}, \mathit{skip})} \lisymb \emptylist}{\emptydict}
    \TStrans{S} \cdots
  \end{array}\]
\end{example}

The next definition is at the core of our proof. It is the
total-correctness counterpart to
Definition~\ref{def:partial-correctness}.

\begin{definition}[Total Correctness]
  \label{def:total-correctness}
  We say that an all-path reachability rule $\varphi \aptrans
  \varphi'$ is \emph{totally satisfied} by $S$, denoted by $S
  \models_t \varphi \aptrans \varphi'$, iff for all complete or
  diverging $\TStrans{S}$-executions $\tau$ starting with $\gamma \in
  \mathcal{T}_{\Cfg}$ and for all $\rho : \Var \rightarrow
  \mathcal{T}$ such that $(\gamma, \rho) \models \varphi$, there
  exists some $\gamma' \in \tau$ such that $(\gamma', \rho) \models
  \varphi'$.
\end{definition}

We now discuss how the definition above generalizes the usual
definition for total correctness found in the literature.

A Hoare tuple $\{ \phi \} P \{ \phi' \}$ is valid in the sense of
total correctness if the precondition $\phi$
entails \begin{enumerate} \item the termination of the program $P$,
  and also \item that the postcondition $\phi'$ holds after the
  program $P$ terminates. \end{enumerate}

Our definition of $S \models_t \varphi \Rightarrow^\forall \varphi'$
states that any execution starting from $\varphi$, terminating or not,
reaches at some point $\varphi'$. If we choose $\varphi'$ to be a
configuration that is known to terminate (e.g., for the case of IMP,
$\cfg{\mathit{skip} \lisymb \emptylist}{\ldots}$), then it follows
that $\varphi$ must terminate along all paths. Otherwise, any
nonterminating path starting with $\varphi$ would meet $\varphi'$,
which terminates, leading to a contradiction.

In particular, the total correctness of the Hoare tuple $\{ \phi \} P \{ \phi'\}$ is encoded by $S \models_t P
\land \phi \Rightarrow^\forall \cfg{\mathit{skip} \lisymb
  \emptylist}{\mathit{env}} \land \phi'$. In addition to encoding
total correctness Hoare tuples, our definition of total correctness is
strictly more general, since it guarantees that $\varphi'$ is reached
in a finite number of steps from $\varphi$, even if $\varphi$ does not
terminate.

We now present our transformation $\ext$, which helps reduce total
correctness guarantees of the form $S \models_t \varphi \aptrans
\varphi'$ to partial correctness sequents of the form $\ext(S) \vdash
\ext(\varphi, s) \aptrans \exists M . \ext(\varphi', M)$, where $\ext$ transforms its
arguments as explained in Theorem~\ref{th:main} below.

\begin{definition}[Reduction From Total Correctness to Partial
  Correctness]%
  \label{def:transformation}
  We define several homonymous maps $\ext$ that encode our
  transformation for reducing total correctness to partial
  correctness. By $\mathcal{S}_{\Sigma}$ we denote the class of all
  algebraic signatures, by $\mathcal{S}$ the class of all sorts and by
  $\mathcal{U}$ the class of all algebras with distinguished sets of
  configurations.
  
\end{definition}

\begin{enumerate}
  
\item \emph{Transforming signatures} \hfill ($\ext : (\mathcal{S}_{\Sigma} \times \mathcal{S}) \rightarrow
  (\mathcal{S}_{\Sigma} \times \mathcal{S})$)

  Let $\Sigma = (S, F)$ be an algebraic signature and $\Cfg \in S$.
  We define $\ext(\Sigma, \Cfg) = (\Sigma', \Cfg'),$ where $\Sigma'
  = (S \cup \Nat \cup \Cfg', \allowbreak F \cup \left\{F_{(), \Nat},
  F_{(\Cfg, \Nat), \Cfg'}, \allowbreak F_{(\Nat, \Nat), \Nat}\right\})$ and where
  $F_{(), \Nat} = \left\{0,1,2, \ldots, \right\}, \allowbreak F_{(\Cfg, \Nat), \Cfg'}
  = \left\{(\:,)\right\}, \allowbreak F_{(\Nat, \Nat), \Nat} =
  \left\{+,-,\times,/\right\}.$
	
  Intuitively, $\theta$ adds a sort for the set of naturals and
  changes the configuration sort such that new configurations consist
  of old configurations, plus a natural number. The natural
  intuitively represents a program variant that is added to the
  configuration, i.e. the maximum number of steps the program can take
  before ending its execution. In addition to the standard operations
  $+, -, \times, /$, we may also consider other operations like
  $|\cdot| : \mathit{Int} \to \Nat$ (absolute value) that operate on
  $\Nat$ and other existing sorts. Alternatively, we could consider
  any well-founded set instead of the set of naturals; however,
  naturals make the presentation easier to follow.

\item \emph{Transforming algebras} \hfill ($\ext : \mathcal{U}
  \rightarrow \mathcal{U}$)

  Let $\mathcal{A} = (A, I_A)$ be a $\Sigma$-algebra, where $\Cfg$ is
  the distinguished sort of configurations and assume $\ext(\Sigma,
  \Cfg) = (\Sigma', \Cfg')$. Then $\ext(\mathcal{A}) = (A', I'_A)$ is
  a $\Sigma'$-algebra with a distinguished sort $\Cfg'$ defined as
  follows:
  \begin{enumerate}
  \item $A \subseteq A'$;
  \item $\mathbb{N} = A'_{\Nat} \in A'$;
  \item $I'_A$ is an extension of $I_A$ such that $\mathcal{T}(n) =
    n_{\mathbb{N}}$, $\mathcal{T}(a \delta b) = \mathcal{T}(a
    \delta_{\mathbb{N}} b)$ for $\delta \in \left\{+, -, \times,
    /\right\}$;
  \item $\mathcal{T}(\Cfg') = \mathcal{T}(\Cfg) \times \mathbb{N}$.
  \end{enumerate}
  Intuitively, each $\Sigma$-algebra with a distinguished sort $\Cfg$
  of configurations is transformed into a $\Sigma$-algebra with a
  distinguished sort $\Cfg'$. The sort $\Cfg'$ is interpreted as pairs
  of old configurations and naturals.

\item \emph{Transforming patterns (matching logic formulae)} \hfill ($\ext :
  (\mathcal{P}_{\mathcal{T}} \times \Term_{\Sigma,\Nat}(\Var))
  \rightarrow \mathcal{P}_{\mathcal{\ext(T)}}$)

  Consider a $\Sigma$-algebra $\mathcal{T}$. Let $\varphi$ be a
  pattern over $\mathcal{T}$ and $n \in \Term_{\Sigma,\Nat}(\Var)$
  ($n$ is a term of sort $\Nat$). We define $\ext$ as follows:
  \begin{enumerate}
  \item if $\varphi$ is structureless, then $\ext(\varphi, n) =
    \varphi$;
  \item if $\varphi$ is a basic pattern, then $\ext(\varphi, n) =
    (\varphi, n)$ (note that this is the interesting case, as in the
    other cases the transformation $\ext$ simply applies
    homomorphically);
  \item if $\varphi = (\varphi_1 \: \delta \: \varphi_2)$ and
    $\varphi$ is not structureless, then $\ext(\varphi, n) =
    \ext(\varphi_1, n) \: \delta \: \ext(\varphi_2, n)$, for $\delta
    \in \left\{ \vee, \wedge \right\}$;
  \item if $\varphi = \delta X(\varphi')$ and $\varphi$ is not
    structureless, then $\ext(\varphi, n) = \delta X \ext(\varphi',
    n)$, for $\delta \in \left\{ \exists, \forall \right\}$;
  \item if $\varphi = \neg\varphi'$ and $\varphi$ is not
    structureless, then $\ext(\varphi, n) = \neg\ext(\varphi', n)$.
  \end{enumerate}
  Intuitively, $\theta$ transforms each old basic pattern into a new
  basic pattern by adding the natural $n$ and ``propagates'' this change
  for all basic patterns contained in the given pattern.
  
\item \emph{Transforming one-path reachability rules} \hfill ($\ext :
  (\mathcal{P}_{\mathcal{T}} \times \mathcal{P}_{\mathcal{T}})
  \rightarrow (\mathcal{P}_{\mathcal{\ext(T)}} \times
  \mathcal{P}_{\mathcal{\ext(T)}})$)

  Let $\varphi \Rightarrow^\exists \varphi'$ be a reachability
  rule. Then $\ext(\varphi \Rightarrow^\exists \varphi') =
  \ext(\varphi, n) \Rightarrow^\exists \ext(\varphi', n - 1)$, where
  $n$ is a fresh variable of sort $\mathit{Nat}$. The transformation
  forces each rule to decrease the program variant (by $1$).
  
\item \emph{Transforming language semantics} \hfill ($\ext :
  2^{(\mathcal{P}_{\mathcal{T}} \times \mathcal{P}_{\mathcal{T}})}
  \rightarrow 2^{(\mathcal{P}_{\mathcal{\ext(T)}} \times
    \mathcal{P}_{\mathcal{\ext(T)}})}$)

  We define the transformation by $\ext(S) = \left\{ \ext(\varphi
  \Rightarrow^\exists \varphi') \spaceVert (\varphi
  \Rightarrow^\exists \varphi') \in S \right\}.$ Each one-path
  reachability rule is transformed independently.
  
\end{enumerate}

We now reduce the problem of total correctness to partial
correctness. This is achieved by the following property of the
transformation $\theta$ defined previously:

\begin{theorem}
  \label{th:main}
  If there exists some term $s \in \Term_{\Sigma, \Nat}(\Var)$
  of sort $\mathit{Nat}$ such that \[\ext(S) \models
  \ext(\varphi, s) \aptrans \exists M.\ext(\varphi', M),\] where
  $M \in \Var_\Nat$, then  $S \models_t \varphi \aptrans \varphi'.$
\end{theorem}
\begin{proof}
	
  Suppose there exist some valuation $\rho : \Var \rightarrow
  \mathcal{\ext(T)}$, some configuration $\gamma \in
  \mathcal{T}_{\Cfg}$ with the property $(\gamma, \rho) \models
  \varphi$ and a complete or diverging $\TStrans{S}$-execution $\tau =
  \gamma \TStrans{S} \gamma_1 \TStrans{S} \cdots $ such that there is
  no $\gamma'$ in $\tau$ for which $(\gamma', \rho) \models \varphi'$.
  Let $n = \rho(s)$. As $\ext(S) \models \ext(\varphi, s) \aptrans
  \exists M.\ext(\varphi', M)$, we have, by definition, that for all
  complete $\ETStrans{\ext(S)}$-paths $\tau^{\theta} =
  (\gamma^{\theta}, n) \ETStrans{\ext(S)} (\gamma_1^{\theta}, n - 1)
  \ETStrans{\ext(S)} \cdots \ETStrans{\ext(S)} (\gamma_k^{\theta}, n -
  k)$ such that $((\gamma^{\theta}, n), \rho) \models \ext(\varphi,
  s)$, there exists some $(\gamma_p^{\theta}, n - p)$ in
  $\tau^{\theta}$ such that $((\gamma_p^{\theta}, n - p), \rho)
  \models \exists M.\ext(\varphi', M)$.
	
  We distinguish two cases. First, suppose $\tau$ is complete and has
  at most $n$ steps. Consider the path $\tau^{\theta} = (\gamma, n)
  \ETStrans{\ext(S)} (\gamma_1, n - 1) \ETStrans{\ext(S)} \cdots
  \ETStrans{\ext(S)} (\gamma_k, n - k)$. It is easy to see that since
  $\tau$ has at most $n$ steps, $\tau^{\theta}$ is indeed a valid
  $\ETStrans{\ext(S)}$-path. Moreover, since $\tau$ is complete, it is
  easy to see that $\tau^{\theta}$ is also complete. It follows that
  there exists some $(\gamma_p, n - p)$ in $\tau^{\theta}$ such that
  $((\gamma_p, n - p), \rho) \models \exists M.\ext(\varphi',
  M)$. By the definition of satisfaction, this statement implies that
  $(\gamma_p, \rho) \models \varphi'$ and we have obtained a
  contradiction.
	
  For the second case, we have that $\tau$ has more than $n$
  steps. Consider the prefix of $\tau$ of $n$ steps: $\tau' = \gamma
  \TStrans{S} \gamma_1 \TStrans{S} \cdots \TStrans{S}
  \gamma_n$. Consider the $\ETStrans{\ext(S)}$-path $\tau'' = (\gamma,
  n) \ETStrans{\ext(S)} (\gamma_1, n - 1) \ETStrans{\ext(S)} \cdots
  \ETStrans{\ext(S)} (\gamma_n, 0)$. Note that $\tau''$ is indeed a
  valid path in $\ext(S)$ and additionally $((\gamma, n), \rho)
  \models \ext(\varphi, s)$. Moreover, $\tau''$ is complete since
  $(\gamma_n, 0)$ cannot advance in $\ETStrans{\ext(S)}$.
	
  This means that $((\gamma_p, n - p), \rho) \models \exists
  M.\ext(\varphi', M)$ for some value $p$. It is easy to see from the
  definition of satisfaction that this last statement implies
  $(\gamma_p, \rho) \models \varphi'$. Since $\ext(\gamma_p, 0)$ is in
  $\tau''$, then $\gamma_p$ is in $\tau'$, which obviously implies
  that $\gamma_p$ is in $\tau$ as well. Therefore, there exists
  $\gamma_p$ in $\tau$ for which $(\gamma_p, \rho) \models \varphi'$.

  We have arrived at a contradiction in both cases, from which we draw
  the conclusion that for all complete or diverging
  $\TStrans{S}$-paths $\tau$ starting with $\gamma \in
  \mathcal{T}_{\Cfg}$ such that $(\gamma, \rho) \models \varphi$,
  there exists some $\gamma'$ in $\tau$ such that $(\gamma', \rho)
  \models \varphi'$. By definition, this means that $S \models(\varphi
  \aptrans_t \varphi')$, which is what we had to prove.
\end{proof}

\begin{corollary}
	If there exists $s \in \Term_{\Sigma, \Nat}(\Var)$
        of sort $\mathit{Nat}$ such that $\ext(S) \models
        \ext(\varphi, s) \aptrans \exists M.\ext(\varphi', M)$, where
        $M \in \Var_\Nat$, then:
	\begin{enumerate}
		\item $S \models \varphi \aptrans \varphi'$;
		\item If $\varphi'$ terminates in $S$, then $\varphi$ also terminates in $S$.
	\end{enumerate}
\end{corollary}

The converse of the corollary above, stating that if a partial
correctness guarantee holds and $\varphi$ terminates then the total
correctness guarantee holds as well, in the cases of
finitely-branching transition systems (this is an immediate
consequence of König's lemma). Given the program SUM in our running
example and the semantics $S$ of IMP, the following sequent can be
derived:

\[\begin{array}{l} \ext(S) \vdash (\cfg{\textit{SUM}}{\mathit{env}_1}, 200|z| + 200) \land
    \mathit{lookup}(m, \mathit{env}_1) = z \land z \geq 0
    \Rightarrow^{\forall} \\ \qquad \exists
    M,\mathit{env}_2.((\cfg{\mathit{skip}}{\mathit{env}_2}, M)\land
    \mathit{lookup}(s, \mathit{env}_2) = z(z+1)/2),\end{array}\]
\noindent which proves the total correctness of \textit{SUM}. A fully
worked out example of a proof of total correctness is given in
Appendix~\ref{app:full-example}.

\section{Related Work}

\label{sec:related}

We critically rely on previous work on language-parametric partial
program correctness, as developed
in~\cite{stefanescu-ciobaca-mereuta-moore-serbanuta-rosu-2014-rta}. Starting
with the operational semantics of the language of the program for
which we prove total correctness, we transform it into an (artificial)
language whose configurations consist of the configurations of the
initial language, plus a variant. This construction is
automated. Given a program and a program variant, its total
correctness in the original language reduces to showing partial
correctness in the new language. Language transformations have been
used before, for example to develop language-parametric symbolic
execution engines~\cite{JSC2016} or language-parametric partial
equivalence checkers~\cite{DBLP:conf/synasc/Ciobaca14}.

In general, the research community treats the subject of termination
orthogonally to the subject of partial correctness. There are several
automated approaches to proving (and certifying) termination
(e.g.,~\cite{Giesl2017,ceta,contejean07frocos,ALARCON2007105}), but
these are typically only concerned with termination, and not
correctness. Therefore, to establish total correctness we generally
first establish partial correctness by using various Hoare-like logics
(e.g.,~\cite{Cao2018,stefanescu-park-yuwen-li-rosu-2016-oopsla}), and
then termination using a specialized termination prover
(e.g.,~\cite{Giesl2017,ALARCON2007105}).

Logics that prove total correctness directly
(e.g.,~\cite{key-ifm2017,da2016modular}) are used more rarely. This is
despite the fact that relatively recent work in automated termination
proving
(e.g.,~\cite{termination-cav2013,variance-popl2007,ramsey-tacas2013,t2-tacas2016,terminator-pldi2006})
shows that it is beneficial to use information obtained by proving a
program (e.g., invariants) in the termination argument:
in~\cite{termination-cav2013}, a cooperation graph is used to enable
the cooperation between a safety prover and the rank synthesis tool,
in~\cite{variance-popl2007}, a variance analysis is introduced that is
parametric in an invariance analysis and Ramsey-based termination
arguments are improved with lexicographic ordering
in~\cite{ramsey-tacas2013}.

\section{Conclusion and Future Work}

\label{sec:conclusion}

We have developed a language semantics transformation that can be used
to prove total correctness of programs. The method can be used for any
programming language whose operational semantics is given by a set of
reachability rules. This is not a restriction, as any programming
language~\cite{DBLP:journals/iandc/SerbanutaRM09} can be faithfully
encoded as such. Moreover, our definition of total correctness
(Definition~\ref{def:total-correctness}) generalizes the usual
definition of total correctness, as it can also be used to reason
about nonterminating programs that are guaranteed to reach a desired
configuration (which could be nonterminating) in a finite number of
steps. We have implemented our approach in the RMT
tool~\cite{rmt,ijcar-2018}. Instructions on obtaining RMT are
available at at
\url{http://profs.info.uaic.ro/~stefan.ciobaca/wpte2018}, along with
several examples for total correctness (including our running
example). Our examples show that our approach works in practice, but
in future work we must also benchmark realistic languages with
reachability logic semantics such as C
(see~\cite{ellison-rosu-2012-popl}) or Java
(see~\cite{DBLP:conf/popl/BogdanasR15}). A limitation of our approach
is that the number of steps has to be computable upfront. This means
that we cannot handle programs that nondeterministically choose a
value and loop for that number of steps. Another limitation is that
the upper bound is not found automatically (even in simple cases), it
has to be provided by the user.

There remain many exciting open questions for future work. The main
question is proving our reduction to be complete. We will also study
how our notion of total correctness corresponds to the well-known
notions of may-convergence and must-convergence in the literature on
process algebra (e.g.,in~\cite{schmidt2010closures}). Another open
question is whether our generalization of the notion of total
correctness has any practical advantages over the usual definition. In
our present approach, the program variant must be a natural number,
but an important question is to analyze whether other well-founded
orders could be needed as well. Another open question is
compositionality: instead of providing a program-wide variant, would
it be possible to have a more modular approach? Finally, can we
combine our method with existing state of the art automated
termination provers like~\cite{t2-tacas2016,terminator-pldi2006} to
obtain the benefits of both?

\section*{Acknowledgement}

This work is funded by the Ministry of Research and Innovation within
Program 1 – Development of the national RD system, Subprogram 1.2 –
Institutional Performance – RDI excellence funding projects, Contract
no.34PFE/19.10.2018.

\bibliographystyle{eptcs}
\bibliography{refs}

\begin{thebibliography}{10}
\providecommand{\bibitemdeclare}[2]{}
\providecommand{\surnamestart}{}
\providecommand{\surnameend}{}
\providecommand{\urlprefix}{Available at }
\providecommand{\url}[1]{\texttt{#1}}
\providecommand{\href}[2]{\texttt{#2}}
\providecommand{\urlalt}[2]{\href{#1}{#2}}
\providecommand{\doi}[1]{doi:\urlalt{http://dx.doi.org/#1}{#1}}
\providecommand{\bibinfo}[2]{#2}

\bibitemdeclare{article}{ALARCON2007105}
\bibitem{ALARCON2007105}
\bibinfo{author}{Beatriz \surnamestart Alarcon\surnameend},
  \bibinfo{author}{Raul \surnamestart Gutierrez\surnameend},
  \bibinfo{author}{Jose \surnamestart Iborra\surnameend} \&
  \bibinfo{author}{Salvador \surnamestart Lucas\surnameend}
  (\bibinfo{year}{2007}): \emph{\bibinfo{title}{Proving Termination of
  Context-Sensitive Rewriting with {MU-TERM}}}.
\newblock {\sl \bibinfo{journal}{ENTCS}} \bibinfo{volume}{188}, pp.
  \bibinfo{pages}{105 -- 115}, \doi{10.1016/j.entcs.2007.05.041}.

\bibitemdeclare{inproceedings}{ceta}
\bibitem{ceta}
\bibinfo{author}{Martin \surnamestart Avanzini\surnameend},
  \bibinfo{author}{Christian \surnamestart Sternagel\surnameend} \&
  \bibinfo{author}{Ren{\'e} \surnamestart Thiemann\surnameend}
  (\bibinfo{year}{2015}): \emph{\bibinfo{title}{{Certification of Complexity
  Proofs using CeTA}}}.
\newblock In: {\sl \bibinfo{booktitle}{RTA}}, {\sl
  \bibinfo{series}{LIPIcs}}~\bibinfo{volume}{36}, pp. \bibinfo{pages}{23--39},
  \doi{10.4230/LIPIcs.RTA.2015.23}.
\newblock \urlprefix\url{http://drops.dagstuhl.de/opus/volltexte/2015/5187}.

\bibitemdeclare{inproceedings}{variance-popl2007}
\bibitem{variance-popl2007}
\bibinfo{author}{Josh \surnamestart Berdine\surnameend}, \bibinfo{author}{Aziem
  \surnamestart Chawdhary\surnameend}, \bibinfo{author}{Byron \surnamestart
  Cook\surnameend}, \bibinfo{author}{Dino \surnamestart Distefano\surnameend}
  \& \bibinfo{author}{Peter \surnamestart O'Hearn\surnameend}
  (\bibinfo{year}{2007}): \emph{\bibinfo{title}{Variance Analyses from
  Invariance Analyses}}.
\newblock In: {\sl \bibinfo{booktitle}{POPL}}, pp. \bibinfo{pages}{211--224},
  \doi{10.1145/1190216.1190249}.

\bibitemdeclare{inproceedings}{DBLP:conf/popl/BogdanasR15}
\bibitem{DBLP:conf/popl/BogdanasR15}
\bibinfo{author}{Denis \surnamestart Bogd\u{a}na\c{s}\surnameend} \&
  \bibinfo{author}{Grigore \surnamestart Ro\c{s}u\surnameend}
  (\bibinfo{year}{2015}): \emph{\bibinfo{title}{{K}-{J}ava: A Complete
  Semantics of {J}ava}}.
\newblock In: {\sl \bibinfo{booktitle}{{POPL}}}, pp. \bibinfo{pages}{445--456},
  \doi{10.1145/2676726.2676982}.

\bibitemdeclare{inproceedings}{termination-cav2013}
\bibitem{termination-cav2013}
\bibinfo{author}{Marc \surnamestart Brockschmidt\surnameend},
  \bibinfo{author}{Byron \surnamestart Cook\surnameend} \&
  \bibinfo{author}{Carsten \surnamestart Fuhs\surnameend}
  (\bibinfo{year}{2013}): \emph{\bibinfo{title}{Better Termination Proving
  through Cooperation}}.
\newblock In: {\sl \bibinfo{booktitle}{CAV}}, pp. \bibinfo{pages}{413--429},
  \doi{10.1007/978-3-642-39799-8\_28}.

\bibitemdeclare{inproceedings}{t2-tacas2016}
\bibitem{t2-tacas2016}
\bibinfo{author}{Marc \surnamestart Brockschmidt\surnameend},
  \bibinfo{author}{Byron \surnamestart Cook\surnameend}, \bibinfo{author}{Samin
  \surnamestart Ishtiaq\surnameend}, \bibinfo{author}{Heidy \surnamestart
  Khlaaf\surnameend} \& \bibinfo{author}{Nir \surnamestart Piterman\surnameend}
  (\bibinfo{year}{2016}): \emph{\bibinfo{title}{{T2}: Temporal Property
  Verification}}.
\newblock In: {\sl \bibinfo{booktitle}{TACAS}}, pp. \bibinfo{pages}{387--393},
  \doi{10.1007/978-3-662-49674-9\_22}.

\bibitemdeclare{article}{Cao2018}
\bibitem{Cao2018}
\bibinfo{author}{Qinxiang \surnamestart Cao\surnameend},
  \bibinfo{author}{Lennart \surnamestart Beringer\surnameend},
  \bibinfo{author}{Samuel \surnamestart Gruetter\surnameend},
  \bibinfo{author}{Josiah \surnamestart Dodds\surnameend} \&
  \bibinfo{author}{Andrew~W. \surnamestart Appel\surnameend}
  (\bibinfo{year}{2018}): \emph{\bibinfo{title}{{VST-Floyd}: A Separation Logic
  Tool to Verify Correctness of {C} Programs}}.
\newblock {\sl \bibinfo{journal}{JAR}}, \doi{10.1007/s10817-018-9457-5}.

\bibitemdeclare{inproceedings}{ijcar-2018}
\bibitem{ijcar-2018}
\bibinfo{author}{\surnamestart \c{S}tefan Ciob\^ac\u{a}\surnameend} \&
  \bibinfo{author}{Dorel \surnamestart Lucanu\surnameend}
  (\bibinfo{year}{2018}): \emph{\bibinfo{title}{A Coinductive Approach to
  Proving Reachability Properties in Logically Constrained Term Rewriting
  Systems}}.
\newblock In: {\sl \bibinfo{booktitle}{{IJCAR}}}, pp.
  \bibinfo{pages}{295--311}, \doi{10.1007/978-3-319-94205-6\_20}.

\bibitemdeclare{inproceedings}{DBLP:conf/synasc/Ciobaca14}
\bibitem{DBLP:conf/synasc/Ciobaca14}
\bibinfo{author}{{\c{S}}tefan \surnamestart Ciob\^ac\u{a}\surnameend}
  (\bibinfo{year}{2014}): \emph{\bibinfo{title}{Reducing Partial Equivalence to
  Partial Correctness}}.
\newblock In: {\sl \bibinfo{booktitle}{{SYNASC}}}, pp.
  \bibinfo{pages}{164--171}, \doi{10.1109/SYNASC.2014.30}.

\bibitemdeclare{techreport}{rmt}
\bibitem{rmt}
\bibinfo{author}{{\c S}tefan \surnamestart Ciob\^ac\u{a}\surnameend} \&
  \bibinfo{author}{Dorel \surnamestart Lucanu\surnameend}
  (\bibinfo{year}{2016}): \emph{\bibinfo{title}{{RMT}: Proving Reachability
  Properties in Constrained Term Rewriting Systems Modulo Theories}}.
\newblock \bibinfo{type}{Technical Report} \bibinfo{number}{TR 16-01},
  \bibinfo{institution}{Alexandru Ioan Cuza University, Faculty of Computer
  Science}.

\bibitemdeclare{inproceedings}{contejean07frocos}
\bibitem{contejean07frocos}
\bibinfo{author}{Evelyne \surnamestart Contejean\surnameend},
  \bibinfo{author}{Pierre \surnamestart Courtieu\surnameend},
  \bibinfo{author}{Julien \surnamestart Forest\surnameend},
  \bibinfo{author}{Olivier \surnamestart Pons\surnameend} \&
  \bibinfo{author}{Xavier \surnamestart Urbain\surnameend}
  (\bibinfo{year}{2007}): \emph{\bibinfo{title}{Certification of Automated
  Termination Proofs}}.
\newblock In: {\sl \bibinfo{booktitle}{{FroCoS}}}, pp.
  \bibinfo{pages}{148--162}, \doi{10.1007/978-3-540-74621-8\_10}.

\bibitemdeclare{inproceedings}{terminator-pldi2006}
\bibitem{terminator-pldi2006}
\bibinfo{author}{Byron \surnamestart Cook\surnameend}, \bibinfo{author}{Andreas
  \surnamestart Podelski\surnameend} \& \bibinfo{author}{Andrey \surnamestart
  Rybalchenko\surnameend} (\bibinfo{year}{2006}):
  \emph{\bibinfo{title}{Termination Proofs for Systems Code}}.
\newblock In: {\sl \bibinfo{booktitle}{PLDI}}, pp. \bibinfo{pages}{415--426},
  \doi{10.1145/1133981.1134029}.

\bibitemdeclare{inproceedings}{ramsey-tacas2013}
\bibitem{ramsey-tacas2013}
\bibinfo{author}{Byron \surnamestart Cook\surnameend}, \bibinfo{author}{Abigail
  \surnamestart See\surnameend} \& \bibinfo{author}{Florian \surnamestart
  Zuleger\surnameend} (\bibinfo{year}{2013}): \emph{\bibinfo{title}{Ramsey vs.
  Lexicographic Termination Proving}}.
\newblock In: {\sl \bibinfo{booktitle}{TACAS}}, pp. \bibinfo{pages}{47--61},
  \doi{10.1007/978-3-642-36742-7\_4}.

\bibitemdeclare{inproceedings}{stefanescu-ciobaca-mereuta-moore-serbanuta-rosu-2014-rta}
\bibitem{stefanescu-ciobaca-mereuta-moore-serbanuta-rosu-2014-rta}
\bibinfo{author}{Andrei \surnamestart \c{S}tef\u{a}nescu\surnameend},
  \bibinfo{author}{{\c{S}}tefan \surnamestart Ciob\^{a}c\u{a}\surnameend},
  \bibinfo{author}{Radu \surnamestart Mereu\c{t}\u{a}\surnameend},
  \bibinfo{author}{Brandon~M. \surnamestart Moore\surnameend},
  \bibinfo{author}{Traian~Florin \surnamestart
  \c{S}erb\u{a}nu\c{t}\u{a}\surnameend} \& \bibinfo{author}{Grigore
  \surnamestart Ro\c{s}u\surnameend} (\bibinfo{year}{2014}):
  \emph{\bibinfo{title}{All-Path Reachability Logic}}.
\newblock In: {\sl \bibinfo{booktitle}{RTA-TLCA}}, pp.
  \bibinfo{pages}{425--440}, \doi{10.1007/978-3-319-08918-8\_29}.

\bibitemdeclare{inproceedings}{stefanescu-park-yuwen-li-rosu-2016-oopsla}
\bibitem{stefanescu-park-yuwen-li-rosu-2016-oopsla}
\bibinfo{author}{Andrei \surnamestart \c{S}tef\u{a}nescu\surnameend},
  \bibinfo{author}{Daejun \surnamestart Park\surnameend},
  \bibinfo{author}{Shijiao \surnamestart Yuwen\surnameend},
  \bibinfo{author}{Yilong \surnamestart Li\surnameend} \&
  \bibinfo{author}{Grigore \surnamestart Ro\c{s}u\surnameend}
  (\bibinfo{year}{2016}): \emph{\bibinfo{title}{Semantics-Based Program
  Verifiers for All Languages}}.
\newblock In: {\sl \bibinfo{booktitle}{OOPSLA}}, pp. \bibinfo{pages}{74--91},
  \doi{10.1145/2983990.2984027}.

\bibitemdeclare{inproceedings}{ellison-rosu-2012-popl}
\bibitem{ellison-rosu-2012-popl}
\bibinfo{author}{Chucky \surnamestart Ellison\surnameend} \&
  \bibinfo{author}{Grigore \surnamestart Ro{\c s}u\surnameend}
  (\bibinfo{year}{2012}): \emph{\bibinfo{title}{An Executable Formal Semantics
  of {C} with Applications}}.
\newblock In: {\sl \bibinfo{booktitle}{{POPL}}}, pp. \bibinfo{pages}{533--544},
  \doi{10.1145/2103656.2103719}.

\bibitemdeclare{article}{Giesl2017}
\bibitem{Giesl2017}
\bibinfo{author}{J{\"u}rgen \surnamestart Giesl\surnameend},
  \bibinfo{author}{Cornelius \surnamestart Aschermann\surnameend},
  \bibinfo{author}{Marc \surnamestart Brockschmidt\surnameend},
  \bibinfo{author}{Fabian \surnamestart Emmes\surnameend},
  \bibinfo{author}{Florian \surnamestart Frohn\surnameend},
  \bibinfo{author}{Carsten \surnamestart Fuhs\surnameend},
  \bibinfo{author}{Jera \surnamestart Hensel\surnameend},
  \bibinfo{author}{Carsten \surnamestart Otto\surnameend},
  \bibinfo{author}{Martin \surnamestart Pl{\"u}cker\surnameend},
  \bibinfo{author}{Peter \surnamestart Schneider-Kamp\surnameend},
  \bibinfo{author}{Thomas \surnamestart Str{\"o}der\surnameend},
  \bibinfo{author}{Stephanie \surnamestart Swiderski\surnameend} \&
  \bibinfo{author}{Ren{\'e} \surnamestart Thiemann\surnameend}
  (\bibinfo{year}{2017}): \emph{\bibinfo{title}{Analyzing Program Termination
  and Complexity Automatically with {AProVE}}}.
\newblock {\sl \bibinfo{journal}{JAR}}
  \bibinfo{volume}{58}(\bibinfo{number}{1}), pp. \bibinfo{pages}{3--31},
  \doi{10.1007/s10817-016-9388-y}.

\bibitemdeclare{article}{JSC2016}
\bibitem{JSC2016}
\bibinfo{author}{Dorel \surnamestart Lucanu\surnameend}, \bibinfo{author}{Vlad
  \surnamestart Rusu\surnameend} \& \bibinfo{author}{Andrei \surnamestart
  Arusoaie\surnameend} (\bibinfo{year}{2017}): \emph{\bibinfo{title}{A generic
  framework for symbolic execution: a coinductive approach}}.
\newblock {\sl \bibinfo{journal}{J. Symb. Comput.}} \bibinfo{volume}{80}, pp.
  \bibinfo{pages}{125--163}, \doi{10.1016/j.jsc.2016.07.012}.

\bibitemdeclare{inproceedings}{da2016modular}
\bibitem{da2016modular}
\bibinfo{author}{Pedro \surnamestart da~Rocha~Pinto\surnameend},
  \bibinfo{author}{Thomas \surnamestart Dinsdale-Young\surnameend},
  \bibinfo{author}{Philippa \surnamestart Gardner\surnameend} \&
  \bibinfo{author}{Julian \surnamestart Sutherland\surnameend}
  (\bibinfo{year}{2016}): \emph{\bibinfo{title}{Modular Termination
  Verification for Non-Blocking Concurrency}}.
\newblock In: {\sl \bibinfo{booktitle}{ESOP}}, pp. \bibinfo{pages}{176--201},
  \doi{10.1007/978-3-662-49498-1\_8}.

\bibitemdeclare{inproceedings}{rosu-stefanescu-2012-oopsla}
\bibitem{rosu-stefanescu-2012-oopsla}
\bibinfo{author}{Grigore \surnamestart Ro\c{s}u\surnameend} \&
  \bibinfo{author}{Andrei \surnamestart \c{S}tef\u{a}nescu\surnameend}
  (\bibinfo{year}{2012}): \emph{\bibinfo{title}{Checking Reachability using
  Matching Logic}}.
\newblock In: {\sl \bibinfo{booktitle}{OOPSLA}}, pp. \bibinfo{pages}{555--574},
  \doi{10.1145/2384616.2384656}.

\bibitemdeclare{inproceedings}{rosu-stefanescu-ciobaca-moore-2013-lics}
\bibitem{rosu-stefanescu-ciobaca-moore-2013-lics}
\bibinfo{author}{Grigore \surnamestart Ro\c{s}u\surnameend},
  \bibinfo{author}{Andrei \surnamestart \c{S}tef\u{a}nescu\surnameend},
  \bibinfo{author}{{\c S}tefan \surnamestart Ciob\^{a}c\u{a}\surnameend} \&
  \bibinfo{author}{Brandon~M. \surnamestart Moore\surnameend}
  (\bibinfo{year}{2013}): \emph{\bibinfo{title}{One-Path Reachability Logic}}.
\newblock In: {\sl \bibinfo{booktitle}{LICS}}, pp. \bibinfo{pages}{358--367},
  \doi{10.1109/LICS.2013.42}.

\bibitemdeclare{inproceedings}{rosu-ellison-schulte-2010-amast}
\bibitem{rosu-ellison-schulte-2010-amast}
\bibinfo{author}{Grigore \surnamestart Ro\c{s}u\surnameend},
  \bibinfo{author}{Chucky \surnamestart Ellison\surnameend} \&
  \bibinfo{author}{Wolfram \surnamestart Schulte\surnameend}
  (\bibinfo{year}{2010}): \emph{\bibinfo{title}{Matching Logic: An Alternative
  to {H}oare/{F}loyd Logic}}.
\newblock In: {\sl \bibinfo{booktitle}{AMAST}}, {\sl \bibinfo{series}{LNCS}}
  \bibinfo{volume}{6486}, pp. \bibinfo{pages}{142--162},
  \doi{10.1007/978-3-642-17796-5\_9}.

\bibitemdeclare{article}{schmidt2010closures}
\bibitem{schmidt2010closures}
\bibinfo{author}{Manfred \surnamestart Schmidt-Schau{\ss}\surnameend} \&
  \bibinfo{author}{David \surnamestart Sabel\surnameend}
  (\bibinfo{year}{2010}): \emph{\bibinfo{title}{Closures of may-, should-and
  must-convergences for contextual equivalence}}.
\newblock {\sl \bibinfo{journal}{Information Processing Letters}}
  \bibinfo{volume}{110}(\bibinfo{number}{6}), pp. \bibinfo{pages}{232--235},
  \doi{10.1016/j.ipl.2010.01.001}.

\bibitemdeclare{article}{DBLP:journals/iandc/SerbanutaRM09}
\bibitem{DBLP:journals/iandc/SerbanutaRM09}
\bibinfo{author}{Traian~Florin \surnamestart {\c S}erb{\u a}nu{\c t}{\u
  a}\surnameend}, \bibinfo{author}{Grigore \surnamestart Ro{\c s}u\surnameend}
  \& \bibinfo{author}{Jos{\'e} \surnamestart Meseguer\surnameend}
  (\bibinfo{year}{2009}): \emph{\bibinfo{title}{A Rewriting Logic Approach to
  Operational Semantics}}.
\newblock {\sl \bibinfo{journal}{{I}nformation and {C}omputation}}
  \bibinfo{volume}{207}(\bibinfo{number}{2}), pp. \bibinfo{pages}{305--340},
  \doi{10.1016/j.ic.2008.03.026}.

\bibitemdeclare{inproceedings}{key-ifm2017}
\bibitem{key-ifm2017}
\bibinfo{author}{Dominic \surnamestart Steinh{\"o}fel\surnameend} \&
  \bibinfo{author}{Nathan \surnamestart Wasser\surnameend}
  (\bibinfo{year}{2017}): \emph{\bibinfo{title}{A New Invariant Rule for the
  Analysis of Loops with Non-standard Control Flows}}.
\newblock In: {\sl \bibinfo{booktitle}{iFM}}, pp. \bibinfo{pages}{279--294},
  \doi{10.1007/978-3-319-66845-1\_18}.

\bibitemdeclare{misc}{winskel1993formal}
\bibitem{winskel1993formal}
\bibinfo{author}{Glynn \surnamestart Winskel\surnameend}
  (\bibinfo{year}{1993}): \emph{\bibinfo{title}{The formal semantics of
  programming languages. Foundations of Computing}}.

\end{thebibliography}

\clearpage

\appendix

\section{Proof System for Partial Correctness}

\label{app:proofsystem}

We recall in
Figure~\ref{fig:proofsystemrta2014} the proof system for the problem
of partial correctness
from~\cite{stefanescu-ciobaca-mereuta-moore-serbanuta-rosu-2014-rta}.

Matching logic formulae can be translated into FOL formulae such that
matching logic satisfaction reduces to FOL satisfaction in the model
of configurations $\mathcal{T}$. This allows conventional theorem
provers to be used for matching logic reasoning. One of the proof
rules of reachability logic depends on this translation.

\begin{definition}
Let $\square$ be a fresh variable of sort $\Cfg$. For a pattern
$\varphi$, let $\varphi^{\square}$ be the FOL formula formed from
$\varphi$ by replacing basic patterns $\pi \in \Term_{\Sigma,
  \Cfg}(\Var)$ with equalities $\square = \pi$. If $\rho : \Var
\rightarrow \mathcal{T}$ and $\gamma \in \mathcal{T}_{\Cfg}$ then let
the valuation $\rho^{\gamma} : \Var \cup \left\{\square\right\}$ be
such that $\rho^{\gamma}(x) = \rho(x)$ for all $x \in \Var$ and
$\rho^{\gamma}(\square) = \gamma$.
\end{definition}

We have that
\begin{center}
  $(\gamma, \rho) \models \varphi \iff \rho^{\gamma} \models
  \varphi^{\square}$.
\end{center}

We use $\varphi[c/\square]$ to denote the FOL formula resulting from
eliminating $\square$ from $\varphi$ and replacing it with a $\Cfg$
variable $c$.

\begin{figure}[t]
\begin{mathpar}

  \inferrule[Step]
	    {\models \varphi \rightarrow \bigvee\limits_{\varphi_l \Rightarrow^{\exists} \varphi_r \in S} \exists \mathit{FreeVars(\varphi_l)} \varphi_l \\\\
	      \models \exists c (\varphi[c/\square] \wedge \varphi_l[c/\square]) \wedge \varphi_r \rightarrow \varphi' \textrm{ for all } \varphi_l \optrans \varphi_r \in S}
	    {\bigTab \mathcal{S, A} \vdash_\mathcal{C} \varphi \Rightarrow^{\forall} \varphi'}
            
\and
            
  \inferrule[Axiom]
	    {\varphi \aptrans \varphi' \in \mathcal{A}}
	    {\mathcal{S, A} \vdash_\mathcal{C} \varphi \aptrans \varphi'}

\and
            
\inferrule[Transitivity]
	{\mathcal{S, A} \vdash_\mathcal{C} \varphi_1 \aptrans \varphi_2 \\
		\mathcal{S, A} \cup \mathcal{C} \vdash \varphi_2 \aptrans \varphi_3}
	{\mathcal{S, A} \vdash_\mathcal{C} \varphi_1 \aptrans \varphi_3}

\and

\inferrule[Case Analysis]
	{\mathcal{S, A} \vdash_\mathcal{C} \varphi_1 \aptrans \varphi \\
		\mathcal{S, A} \vdash_\mathcal{C} \varphi_2 \aptrans \varphi}
	{\mathcal{S, A} \vdash_\mathcal{C} \varphi_1 \vee \varphi_2 \aptrans \varphi}

\and

\inferrule[Circularity]
	{\mathcal{S, A} \vdash_\mathcal{C \cup \left\{\varphi \aptrans \varphi' \right\} } \varphi \aptrans \varphi'}
	{\mathcal{S, A} \vdash_\mathcal{C} \varphi \aptrans \varphi'}

\and

\inferrule[Abstraction]
	{\mathcal{S, A} \vdash_\mathcal{C} \varphi \aptrans \varphi' \\
		X \cap \mathit{FreeVars}(\varphi') = \emptyset}
	{\mathcal{S, A} \vdash_\mathcal{C} \exists X \varphi \aptrans \varphi'}

\and

\inferrule[Reflexivity]
	{.}
	{\mathcal{S, A} \vdash_\mathcal{C} \varphi \aptrans \varphi}

\and

\inferrule[Consequence]
	{\models \varphi_1 \rightarrow \varphi_1' \\
		\mathcal{S, A} \vdash_\mathcal{C} \varphi_1' \aptrans \varphi_2' \\
		\models \varphi_2' \rightarrow \varphi_2}
	{\mathcal{S, A} \vdash_\mathcal{C} \varphi_1 \aptrans \varphi_2}

\end{mathpar}
\caption{\label{fig:proofsystemrta2014}The language-parametric proof
  system for partial correctness
  in~\cite{stefanescu-ciobaca-mereuta-moore-serbanuta-rosu-2014-rta}}
\end{figure}

The proof system was shown
in~\cite{stefanescu-ciobaca-mereuta-moore-serbanuta-rosu-2014-rta} to
be sound (and also relatively complete) for the problem of partial
correctness. Note that, this provides no guarantees for configurations
that do not terminate.

\section{A Complete Example}

\label{app:full-example}

In this section, we present in full details a very simple example of
how the reduction presented above work. We consider a very simple
``language'' with configurations of the form $[s, i]$ (where $s$ and
$i$ are naturals) that add to $s$ the first $i$ positive naturals.

Let $\Sigma = (\left\{ \Cfg, \Nat, \mathit{Bool} \right\}, F)$, where:

\begin{list}{$\bullet$}{}

\item $F_{(),\Nat} = \left\{0,1,2\ldots\right\}$

\item $F_{(),\mathit{Bool}} = \left\{\mathit{True}, \mathit{False}\right\}$

\item $F_{(\Nat, \Nat), \Nat} = \left\{+,-,/,*\right\}$

\item $F_{(\Nat, \Nat), \mathit{Bool}} = \left\{<,>,\leq,\geq,=\right\}$

\item $F_{(\Nat, \Nat), \Cfg} = \left\{[,]\right\}$

\end{list}

We consider a $\Sigma$-algebra $\mathcal{T}$ with the expected
interpretation for common symbols and a system of reachability rules
$S$ consisting of a single rule:

\begin{center}
  $[s, i] \wedge (i > 0) \Rightarrow [s + i, i - 1]$, where $s, i \in
  \Var_{\Nat}$.
\end{center}

The algebra $\ext(\mathcal{T})$ contains a sort $\Cfg'$ and, by
definition, $\ext(S)$ consists of the following rule:

\begin{center}

  $([s, i], n) \wedge (i > 0) \Rightarrow ([s + i, i - 1], n - 1)$,
  where $s, i, n \in \Var_{\Nat}$.
  
\end{center}

For ease of readability let $\SUM(x, y) = y * (y+1)/2 - (x-1)*x/2$ by
notation. Let $\varphi_L = ([\SUM(n'+1, n), n'], n') \wedge n' \geq 0$
and $\varphi_R = \exists m ([\SUM(1,n), 0], m)$, where $n',n,m \in
\Var_{\Nat}$. Let us now prove that \[\ext(S) \vdash ([0, n], n)
\aptrans \varphi_R,\] which establishes not only that $[0, n]$
computes the sum from $1$ to $n$ (by the soundness of reachability
logic), but also that it terminates within $n$ steps (by
Theorem~\ref{th:main}):

14. $\ext(S), {\left\{\exists n' \varphi_L \aptrans \varphi_R\right\}} \vdash \exists n' \varphi_L \aptrans \varphi_R$\hfill by \textbf{Axiom}

13. $\ext(S), {\left\{\exists n' \varphi_L \aptrans \varphi_R\right\}} \vdash ([\SUM(n', n), n' - 1], n' - 1) \wedge (n' - 1) \geq 0 \aptrans \varphi_R$

\hfill by \textbf{Consequence} from 14

12. $\ext(S) \vdash_{\left\{\exists n' \varphi_L \aptrans \varphi_R\right\}} ([\SUM(n'+1, n), n'], n') \wedge n' > 0 \aptrans$

$\qquad ([\SUM(n', n), n' - 1], n' - 1) \wedge (n' - 1) \geq 0$\hfill by \textbf{Step}

11. $\ext(S) \vdash_{\left\{\exists n' \varphi_L \aptrans \varphi_R\right\}} ([\SUM(n'+1, n), n'], n') \wedge n' > 0 \aptrans \varphi_R$

\hfill by \textbf{Transitivity} from 12 and 13

10. $\ext(S) \vdash_{\left\{\exists n' \varphi_L \aptrans \varphi_R\right\}} ([\SUM(1, n), 0], 0) \aptrans ([\SUM(1,n), 0], 0)$

\hfill by \textbf{Reflexivity}

9. $\ext(S) \vdash_{\left\{\exists n' \varphi_L \aptrans \varphi_R\right\}} ([\SUM(1, n), 0], 0) \aptrans \varphi_R$\hfill by \textbf{Consequence} from 10

8. $\ext(S) \vdash_{\left\{\exists n' \varphi_L \aptrans \varphi_R\right\}} \exists n' (([\SUM(n'+1, n), n'], n') \wedge n' > 0) \aptrans \varphi_R$

\hfill by \textbf{Abstraction} from 11

7. $\ext(S) \vdash_{\left\{\exists n' \varphi_L \aptrans \varphi_R\right\}} \exists (n' ([\SUM(n'+1, n), n'], n') \wedge n' = 0) \aptrans \varphi_R$

\hfill by \textbf{Consequence} from 9

6. $\ext(S) \vdash_{\left\{\exists n' \varphi_L \aptrans \varphi_R\right\}} \exists n' (([\SUM(n'+1, n), n'], n') \wedge n' > 0) \vee$

$\exists n' (([\SUM(n'+1, n), n'], n') \wedge n' = 0) \aptrans \varphi_R$ \hfill by \textbf{Case analysis} from 7, 8

5. $\ext(S) \vdash_{\left\{\exists n' \varphi_L \aptrans \varphi_R\right\}} \exists n' \varphi_L \aptrans \varphi_R$\hfill by \textbf{Consequence} from 6

4. $\ext(S) \vdash ([0, n], n) \aptrans ([0, n], n)$\hfill by \textbf{Reflexivity}

3. $\ext(S) \vdash \exists n' \varphi_L \aptrans \varphi_R$\hfill by \textbf{Circularity} from 5

2. $\ext(S) \vdash ([0, n], n) \aptrans \exists n' \varphi_L$\hfill by \textbf{Consequence} from 4

1. $\ext(S) \vdash ([0, n], n) \aptrans \varphi_R$\hfill by \textbf{Transitivity} from 2 and 3

Our approach also works on the programming language IMP described
above. We have shown, for example, that the following program is
(unsurprisingly) totally correct (when ${\tt m}$ starts up with a
nonnegative number):

\begin{verbatim}
s := 0
while not (m = 0) do s := s + m; m := m - 1
\end{verbatim}

The main idea in proving the program above totally correct is the same
as in the fully developed example above, but the formal proof is a lot
longer.

\end{document}